\documentclass[
article,%
reprint,
 amsmath,amssymb,
 aps,
]{revtex4-2}

\usepackage{xspace}
\usepackage{graphicx}
\usepackage{dcolumn}
\usepackage{bm}
\usepackage{amsthm,wasysym}
\usepackage{url}
\usepackage{multirow}
\usepackage[utf8]{inputenc}
\usepackage[english]{babel}
\usepackage[dvipsnames]{xcolor}
\usepackage[normalem]{ulem}

\newcommand{\Sect}[1]{Section~\ref{#1}}
\newcommand{\Eq}[1]{Equation~\ref{#1}}

\newtheorem{theorem}{Theorem}[section]

\newtheorem{lemma}[theorem]{Lemma}

\theoremstyle{definition}
\newtheorem{definition}{Definition}[section]

\newcommand{\data}{\ensuremath{\vec{d}}}

\newcommand{\pr}[1]{\ensuremath{\mathrm{p}(#1)}}

\newcommand{\gvn}{\mid}

\newcommand{\sz}{spec-$z$}
\newcommand{\pz}{photo-$z$}
\newcommand{\pzpdf}{\pz\ PMF}

\newcommand{\Nz}{$\mathcal{N}(z)$}

\newcommand{\project}[1]{{\textsc{#1}}}
\newcommand{\lsst}{\project{LSST}}

\begin{document}

\title{How not to obtain the redshift distribution from probabilistic redshift estimates: \\
Under what conditions is it not inappropriate to estimate \\
the redshift distribution $N(z)$ by stacking photo-$z$ PDFs?}

\author{Alex I. Malz}
\email{aimalz@astro.ruhr-uni-bochum.de}
\affiliation{Ruhr-University Bochum, German Centre for Cosmological Lensing, Universit\"{a}tsstra{\ss}e 150, 44801 Bochum, Germany}

\date{\today}

\begin{abstract}
The scientific impact of current and upcoming photometric galaxy surveys is contingent on our ability to obtain redshift estimates for large numbers of faint galaxies.
In the absence of spectroscopically confirmed redshifts, broad-band photometric redshift point estimates (photo-$z$s) have been superceded by photo-$z$ probability density functions (PDFs) that encapsulate their nontrivial uncertainties.
Initial applications of photo-$z$ PDFs in weak gravitational lensing studies of cosmology have employed computationally straightforward stacking methodologies for obtaining the redshift distribution function $\mathcal{N}(z)$ that violate the laws of probability.  
In response, mathematically self-consistent models of varying complexity have been proposed in an effort to answer the question, ``What is the right way to obtain the redshift distribution function $\mathcal{N}(z)$ from a catalog of photo-$z$ PDFs?''
This letter aims to motivate adoption of such principled methods by addressing the contrapositive of the more common presentation of such models, answering the question, ``Under what conditions do traditional stacking methods successfully recover the true redshift distribution function $\mathcal{N}(z)$?''
By placing stacking in a rigorous mathematical environment, we identify two such conditions: those of perfectly informative data and perfectly informative prior information. 
Stacking has maintained its foothold in the astronomical community for so long because the conditions in question were only weakly violated in the past.
These conditions, however, will be strongly violated by future galaxy surveys.  
We therefore conclude that stacking must be abandoned in favor of mathematically principled methods in order to advance observational cosmology.
\end{abstract}

\maketitle

\section{Motivation}
\label{sec:intro}

Knowledge of galaxy redshifts $z$  is essential for understanding their evolutionary processes and for utilizing them as cosmological tracers of  large-scale structure.
Spectral observations of absorption and emission lines enable the high-confidence measurement of spectroscopic redshifts (\sz s), however, such measurements are costly and preclude large samples of high-redshift galaxies.
When high-resolution spectra are unavailable, broad-band photometry that is sensitive to the spectral continuum can be used to estimate galaxy redshifts \citep{baum_photoelectric_1962}.
Though inherently noisy, the relationship between redshift and photometry can be established from a library of spectral energy distribution (SED) templates or a representative sample of galaxies with known redshifts.
By boosting the cosmological sample size, especially at high redshift, \pz s have directly enabled the era of precision cosmology derived by weak gravitational lensing tomography and baryon acoustic oscillation peak measurements.  

However, \pz s are by their very nature imprecise and inaccurate compared to their spectroscopic counterparts. 
Per-galaxy \pz\ probability density functions (PDFs), defined over all possible redshifts and usually denoted as $\mathrm{P}(z)$, better encapsulate the nontrivial uncertainty landscape \citep{koo_photometric_1999}.
The most common application of \pz\ PDFs is their use in estimating the redshift distribution function \Nz\  of a sample of galaxies, a quantity essential to cosmological parameter constraints via the power spectra of weak gravitational lensing and large-scale structure \citep{mandelbaum_precision_2008, sheldon_photometric_2012, bonnett_redshift_2016, hildebrandt_kids-450_2017}.

When \pz\ PDFs are available instead of true redshifts $z^{\dagger}$, the simplest approach reduces each $\mathrm{P}(z)$ to a point estimate $\hat{z}$ of redshift by using $\delta(z, \hat{z})$ as a substitute for the unknown (and unknowable) $\delta(z, z^{\dagger})$, where $\delta(z, z') \equiv \{\infty, z = z';\  0, z \neq z'$ is the Dirac delta function between a random variable $z$ and a particular value $z'$ thereof, normalized such that $\int \delta(z, z') dz = 1$.
An intuitive approach to circumventing the loss of valuable knowledge of redshift uncertainty resulting from reduction of \pz\ PDFs to point estimates is to average  their \pz\ PDFs according to the \textit{stacked estimator} of the redshift distribution function \citep{lima_estimating_2008}, 
    \begin{align}
    \label{eqn:stack}
    \hat{\mathcal{N}}(z) &\equiv \sum_{i = 1}^{N} \mathrm{P}(z_{i}) ,
    \end{align}
of a sample of $N$ galaxies $i$.
A number of popular extensions to \Eq{eqn:stack} are addressed in \Sect{sec:math}.

Unfortunately, the stacked estimator $\hat{\mathcal{N}}(z)$ of the redshift distribution is mathematically incorrect \citep{hogg_data_2012}.
Even under simplifying assumptions, stacking yields a biased estimator of $N(z)$.
Consider a scenario in which the true redshift distribution $\mathcal{N}^{\dagger}(z) = \delta(z, z^{\dagger})$ is a delta function at $z^{\dagger}$ and each \pz\ PDF $\mathrm{P}(z_{i}) = \mathrm{N}(\hat{z}_{i}, \sigma^{2})$ is a Gaussian distributions with a shared variance $\sigma^{2}$ and a mean $\hat{z}_{i} \sim \mathrm{N}(z^{\dagger}, \sigma^{2})$ drawn from a Gaussian distribution centered at the shared true redshift $z^{\dagger}$ with the same shared variance $\sigma^{2}$;
in such a situation, $\hat{\mathcal{N}}(z)$ cannot be equal to $\mathcal{N}^{\dagger}(z)$.

Mathematically principled methodologies for recovering the redshift distribution, including those that simultaneously infer the \pz\ PDFs directly from photometry \citep{leistedt_hierarchical_2016, leistedt_hierarchical_2019} and those that use an existing set of \pz\ PDFs \citep{malz_how_2020, rau_composite_2021}, have been proposed and validated.
Despite the availability of robust methodological alternatives, however, \Eq{eqn:stack} and its close cousins have remained prevalent in modern cosmological analysis pipelines. 
In addition to the understandable inertia and expected technical challenges of adopting new approaches, stacking is bolstered by pervasive misconceptions about the causal structure of the problem of redshift inference and \pz\ PDFs as probabilistic objects overall \citep{gruen_combining_2017, jarvis_open_2018, malz_re:_2018}.

This letter is pedagogical in nature and serves as a companion to \citep{malz_how_2020}, which quantifies the detrimental consequences of stacking, derives a mathematically self-consistent alternative methodology for constraining the redshift distribution from a sample of \pz\ PDFs, and provides a public implementation thereof.
Here, we complementarily employ a series of thought experiments to motivate the paradigm shift whose necessity is demonstrated in \citep{malz_how_2020}.
We present a flexible mathematical framework in \Sect{sec:preamble}, derive the general form of the redshift distribution in \Sect{sec:math}, explore limiting cases of the general form in \Sect{sec:results}, and interpret the implications thereof in \Sect{sec:disco}.

\section{Definitions}
\label{sec:preamble}

To answer the question of when stacking can recover the true redshift distribution, we must first dust off our knowledge of probability and explicitly connect these classic concepts to the problem at hand.
We begin by considering the chance that a single galaxy's true redshift $z^{\dagger}_{i}$ takes some reference value $z'$.

\begin{definition}\label{def:binarystatespace}
	The \textit{outcome space} $\Omega(z^{\dagger}_{i})$ of the true redshift $z_{i}^{\dagger}$ of a single-galaxy $i$ is $-\infty < z^{\dagger}_{i} < \infty$, although in cosmology we can safely assume $z \geq 0$.
\end{definition}

\begin{definition}\label{def:event}
	An \textit{event} is a collection of possible outcomes;
	in our case, the experiment may be the natural occurence of the true redshift $z_{i}^{\dagger}$ of a galaxy $i$, though it may be unknown to us.
\end{definition}

\begin{definition}\label{def:disjoint}
	The possible outcomes in $\Omega(z_{i}^{\dagger})$ are \textit{disjoint} if one occurring means that all others cannot occur; 
	our single galaxy $i$ cannot satisfy both ${z^{\dagger}_{i} = z'}$ and ${z^{\dagger}_{i} = \lnot z'}$.
\end{definition}

\begin{definition}\label{def:pdens}
	The \textit{probability density function (PDF)} ${\mathrm{P}}(z^{\dagger}_{i} = z') \geq 0$ is the chance that a galaxy's redshift $z_{i}^{\dagger}$ takes the value $z'$ if $z$ is continuous.
\end{definition}

Though redshift $z$ is continuous in reality, we assume, for the pedagogical purposes of this letter, that redshift $z$ is a discrete random variable ${z_{i}^{\dagger} \in \{z', \lnot z'\}}$ that can take either the value $z'$ of some reference redshift or any other value of redshift $\lnot z'$.
This assumption of a binary discrete outcome space is made without loss of generality, as all results may be straightfowardly extended to the limit of continuous redshift;
guidance to relevant proofs for such extensions is provided throughout \Sect{sec:math}.

\begin{definition}\label{def:pmass}
	In our pedagogical example of the binary discrete outcome space, ${\pr{z^{\dagger}_{i} = z'} \equiv \int_{z'-\epsilon}^{z'+\epsilon} \mathrm{P}(z^{\dagger} = z) \mathrm{d}z \geq 0}$, for a very small $0 \leq \epsilon \ll 1$, is instead a \textit{probability mass function (PMF)};
	it must be noted that this integral over a continuous PDF $\mathrm{P}(z)$ respects the physical units of redshift, whereas the resulting discrete PMF $\pr{z}$ is unitless \cite{hogg_data_2012}.
\end{definition}

\begin{definition}\label{def:normalization}
	A PMF must satisfy the \textit{normalization} condition ${\sum_{z} \pr{z^{\dagger}_{i} = z}} = 1$;
	the PMF over our binary outcome space $\Omega(z^{\dagger}_{i})$ satisfies ${\pr{z^{\dagger}_{i} = z'} = 1 - \pr{z^{\dagger}_{i} = \lnot z'}}$.
\end{definition}

The above definitions pertain to the redshift of a single galaxy $i$, but this letter concerns the distribution \Nz\ of redshifts of an ensemble $S$ of $|S| = N$ galaxies, which is some pre-defined sample that need not include all galaxies in the universe.
A small window about a single redshift $z' \pm \epsilon$ for small $0 < \epsilon \ll 1$ contains the true redshift of some whole number ${\mathcal{N}(z') = K^{\dagger}}$ of those galaxies.
The discrete outcome space ${\Omega(K^{\dagger}) = \{K^{\dagger} = 0, \dots, K^{\dagger} = N\}}$ of the true number $K^{\dagger}$ of galaxies with true redshift $z'$ thus has $N + 1$ elements, and there is some probability ${\pr{K^{\dagger} = K}}$ of each value of $K$ that must be related to the \pzpdf s ${\{\pr{z^{\dagger}_{i} = z'}\}_{i = 1, \dots, N}}$.

\begin{definition}\label{def:conditional}
	We may consider the \textit{conditional} probability that an event occurs given that another event occurs;
	for example, if ${z^{\dagger}_{i} = z'}$, then ${\pr{K^{\dagger} = 0 \gvn z^{\dagger}_{i} = z'} = 0}$.
\end{definition}

\begin{definition}\label{def:independence}
	Multiple events are \textit{statistically independent} if the probability of each occurring does not affect the probability that the others occur;
	for this work, we broadly assume the statistical independence ${\pr{z^{\dagger}_{i}} \perp \pr{z^{\dagger}_{h \neq i}}}$ of the redshifts of individual galaxies $h, i \in S$ from an arbitrary set $S$ of galaxies.
\end{definition}

\begin{definition}\label{def:intersection}
    In a universe with an ensemble $S$ of $|S| = N$ galaxies whose redshifts are statistically independent, the \textit{intersection} of the probability that they all have a specific redshift $z'$, which is also the probability that the number of galaxies with redshift $z'$ is equal to the total number of galaxies ${\pr{K^{\dagger} = N}}$, is the product of the probabilities that each has that redshift: ${\pr{z^{\dagger}_{1} = z' \cap \dots \cap z^{\dagger}_{N} = z'} = \pr{z^{\dagger}_{1} = z'} \times \dots \times \pr{z^{\dagger}_{N} = z'}}$.
\end{definition}

Definition~\ref{def:intersection} implies that the stacked estimator $\hat{\mathcal{N}}(z)$ of the redshift distribution given by Equation~\ref{eqn:stack} can be equal to the true redshift distribution if ${\pr{z^{\dagger}_{1} = z' \cap \dots \cap z^{\dagger}_{N} = z'}} = {\pr{z^{\dagger}_{1} = z'} + \dots + \pr{z^{\dagger}_{N} = z'}}$, which does not directly follow from any combination of the above definitions and in fact may well violate Definition~\ref{def:normalization}.
Because  Equation~\ref{eqn:stack} yields a single estimate $\hat{K}$ rather than the probability distribution $\pr{K}$ over the value of the redshift distribution, we work with a summary statistic of $\pr{K}$ to make an apples to apples comparison.

\begin{definition}\label{def:expected}
	The \textit{expected value} of a continuous random variable such as $z$ is its first moment ${\langle z_{i}^{\dagger} \rangle \equiv \int_{-\infty}^{\infty}z'\mathrm{P}(z_{i}^{\dagger} = z')\mathrm{d}z'}$;
	the expected value of a discrete random variable such as $K$ is defined as ${\langle K^{\dagger} \rangle \equiv \sum_{K' = 0}^{N} K' \pr{K^{\dagger} = K'}}$, which is not restricted to discrete values.
\end{definition}

Now, we may derive the tools necessary to determine when the stacked estimator of the redshift distribution is  equal to the mathematically self-consistent expected value of the redshift distribution and thus valid.

\section{Derivation}
\label{sec:math}

Here we compare the stacked estimator $\hat{\mathcal{N}}(z')$ of the redshift distribution to the expected value of the PMF over all possible values of the redshift distribution, derived using the tools of combinatorics.

We consider a generic set ${S_{j} = \{i_{1}, \dots, i_{K}\}}$ of ${|S_{j}| = K}$ galaxies $i \in S_{j}$ and its complement 
$\lnot S_{j}$ of ${|\lnot S_{j}| = N - K}$ galaxies $i \notin S_{j}$.
All $N$ galaxies belong to one set or the other, and no galaxy can be in both $S_{j}$ and $\lnot S_{j}$.
One can imagine sets such as these being  galaxies with true redshifts equal to $z'$ and $\lnot z'$.

Since galaxies have true redshifts in nature, there is a true number $K^{\dagger}$ of galaxies with $z_{i}^{\dagger} = z'$.
If we aim to recover $\mathcal{N}(z')$, we only care about the size of the set, not its membership.
Thus we seek sets $\{S_{j}\}$ with ${|S_{j}| = K^{\dagger}}$ galaxies out of all $[N]^{K}$ possible subsets of $N$ galaxies.
However, to find the probability that $\mathcal{N}(z')$ takes its true value $K^{\dagger}$, we will need to acknowledge that there is exactly one set $S_{j^{\dagger}}$ out of the set $[N]^{K^{\dagger}}$ of all sets $S_{j}$ with ${|S_{j}| = K^{\dagger}}$ galaxies that contains all galaxies of true redshift ${z_{i}^{\dagger} = z'}$ and no galaxies that have redshift ${z_{i}^{\dagger} = \lnot z'}$, for each possible redshift $z'$.

Constructing all these sets for a cosmologically relevant sample of $N$ galaxies is unfortunately impractical due to the following two theorems, presented here as definitions, whose proofs by induction may be found in \cite{pitman_probability_1999}.

\begin{definition}\label{lem:permutations}
	The number of possible orderings or \textit{permutations} of $N$ galaxies is ${N! \equiv N \cdot (N - 1) \cdot \dots \cdot 1}$.
\end{definition}

\begin{definition}\label{lem:combinations}
	The number of possible unordered subsets or \textit{combinations} of $K$ galaxies that can be selected from a set of $N \geq K$ galaxies is ${\binom{N}{K} \equiv \frac{N!}{K! (N - K)!}}$.
\end{definition}

Because there are so many possible sets of $K^{\dagger}$ galaxies in question and we do not even know the value of $K^{\dagger}$ in the first place, it is not possible to check them by hand, but this thought experiment mandates that we consider them.
We are now prepared to evaluate the generic probability ${\pr{K^{\dagger} = K'}}$ that $\mathcal{N}(z')$ takes any particular value $K'$ at a fixed $z'$, where ${\mathcal{N}(z') \equiv \sum_{i \in S_{j}} \begin{cases} 
1, & z'-\epsilon < z_{i}^{\dagger} < z'+\epsilon\\
0, & \mathrm{else}
\end{cases}}$.
We begin with a toy case of simplifying assumptions and strip them away one by one to obtain a general result.

\subsection{A universe of identical galaxies}
\label{sec:tworedshift}

It is instructive to first consider the special case in which there is one \pzpdf\ ${\pr{z'} \equiv \pr{z_{i}^{\dagger} = z'}}$ shared among each galaxy $i$ in the sample of $N$ galaxies.

\begin{lemma}\label{lem:identical}
	In a universe in which all galaxies have the same \pzpdf\ with probability $\pr{z'}$ of having redshift $z'$, the PMF over the redshift distribution is
	\begin{equation}
	\label{eqn:binomial}
	\pr{K^{\dagger} = K'} = \binom{N}{K'} \pr{z'}^{K'} (1 - \pr{z'})^{N - K'} .
	\end{equation}
\end{lemma}
\begin{proof}
In terms of the discrete binary outcome space of Definition~\ref{def:binarystatespace}, the problem is equivalent to that of counting $K$ successes after flipping a single coin with $\pr{\mathrm{success}} = \pr{z'}$ a total of $N$ times.
The $N$ coin flips build a set $S_{j}$ of ${|S_{j}| = K}$ galaxies with redshift $z'$ and a set $\lnot S_{j}$ of ${|\lnot S_{j}| = N - K}$ galaxies with redshift $\lnot z$'.
If we want the probability of $K'$ successes, we sum over the set $[N]^{K'}$ of all combinations of galaxies that could define these sets.
Bearing in mind that our galaxies have ${\pr{\lnot z'} = 1 - \pr{z'}}$, due to Definitions \ref{def:binarystatespace} and \ref{def:disjoint}, the probability of $K'$ galaxies having redshift $z'$ is given by the binomial theorem, which is \Eq{eqn:binomial}.
\end{proof}

One may generalize the two-redshift universe to a multi-valued redshift universe by following the classic proof extending the binomial theorem to the ``birthday problem'' seeking the probability that $K'$ people in a room of $N$ people share the same birthday $z'$.

To instead estimate the redshift density ${n(z') \equiv \mathcal{N}(z') / N}$, a normalized version of the redshift distribution $\mathcal{N}(z')$, for all galaxies ${N \to \infty}$, not just those observed in a given sample, one may follow the classic derivation of the Poisson distribution from the binomial distribution under the limit of $N\to\infty$.

\subsection{A universe of unique galaxies}
\label{sec:unique}

Of course galaxies do not share the same \pzpdf, so we must next extend the toy case of identical galaxies to a more realistic set of galaxies with unique \pzpdf s ${\{\pr{z'_{i}} \equiv \pr{z_{i}^{\dagger} = z'}\}}$.

\begin{lemma}\label{lem:unique}
	In a universe where galaxies do not necessarily all have the same \pzpdf, the probability that $\mathcal{N}(z')$ takes the value $K'$ is 
	\begin{align}
	\label{eqn:general}
	\pr{K^{\dagger} = K'} &= \sum_{j \backepsilon |S_{j}| = K'}^{\binom{N}{K'}} \left[ \prod_{i \in S_{j}}^{K'} \pr{z'_{i}} \prod_{i \notin S_{j}}^{N - K'} 1 - \pr{z'_{i}} \right] .
	\end{align}
\end{lemma}
\begin{proof}
	By Definition \ref{def:independence}, the probability of obtaining set $S_{j}$ of ${|S_{j}| = K'}$ galaxies is the product of the probabilities that all the galaxies in the set $S_{j}$ have redshift $z'$ and the probabilities that all the galaxies outside that set, in $\lnot S_{j}$, have redshift $\lnot z'$.
	By Definition \ref{def:independence} again, the probability that $S_{j}$ corresponds to the true set of galaxies of redshift $z'$ is the product of the $K'$ probabilities ${\{\pr{z'_{i}}\}_{i \in S_{j}}}$ of the galaxies in $S_{j}$ having redshift $z'$, and the probability that all galaxies outside $S_{j}$ have redshift $\lnot z'$ is the product of the $N - K'$ probabilities ${\{\pr{\lnot z'_{i}}\}_{i \notin S_{j}}}$.
	By Definition \ref{def:binarystatespace}, $\pr{\lnot z'_{i}} = 1 - \pr{z'_{i}}$, so the probability $\pr{S_{j}}$ of obtaining a set of galaxies $S_{j}$ with ${|S_{j}| = K'}$ members is ${\pr{S_{j}} = {\prod_{i \in S_{j}} \pr{z'_{i}} \prod_{i \notin S_{j}} 1 - \pr{z'_{i}}}}$.
	Invoking Definition \ref{def:disjoint}, we know that ${\pr{K^{\dagger} = K'} = \sum_{j \backepsilon |S_{j}| = K'}^{\binom{N}{K'}} \pr{S_{j}}}$.
\end{proof}

Though \Eq{eqn:general} is computationally intractable for any nontrivial number of galaxies (and even moreso in the limit of continuous redshift), we answer the question at hand by comparing it with \Eq{eqn:stack} in the limiting cases of data- and prior-dominated \pzpdf s in Sections~\ref{sec:informative} and \ref{sec:uninformative} below.

\section{Results}
\label{sec:results}

We are now able to derive the general form of the expected value $\langle \mathcal{N}(z')\rangle$ of the number of galaxies with reference redshift $z'$, which we  then compare to the stacked estimator of the redshift distribution $\hat{\mathcal{N}}(z')$ given by Equation~\ref{eqn:stack}.

\begin{theorem}\label{thm:general}
	The expected value of the number of galaxies $\mathcal{N}(z')$ with redshift $z'$ in the general case of a universe of unique galaxies is given by
	\begin{align}
	\label{eqn:complete}
	\langle K^{\dagger} \rangle &= \sum_{K' = 0}^{N} K' \sum_{j \backepsilon |S_{j}| = K'}^{\binom{N}{K'}} \left[ \prod_{i \in S_{j}}^{K'} \pr{z_{i}^{\dagger} = z'} \prod_{i \notin S_{j}}^{N - K'} 1 - \pr{z_{i}^{\dagger} = z'} \right].
	\end{align}
\end{theorem}
\begin{proof}
	Apply Definition \ref{def:expected} to Lemma \ref{lem:unique} to arrive at \Eq{eqn:complete}.
\end{proof}

Equation~\ref{eqn:complete} gives the key quantity against which we may compare the stacked estimator of the redshift distribution given in Equation~\ref{eqn:stack}.
In \Sect{sec:informative} and \Sect{sec:uninformative}, we expose the true nature of \pzpdf s by iteratively correcting our oversimplified notation.
By confronting what ``$\pr{z_{i}}$'' sweeps under the rug, we identify the only conditions under which Equation~\ref{eqn:stack} can yield a result consistent with \Eq{eqn:complete}.

\subsection{Perfectly informative data}
\label{sec:informative}

In Section~\ref{sec:unique}, we assert that \pzpdf s are in general different for each galaxy, and we now justify that claim and consider its implications.
Across the galaxy sample, redshifts are always defined over the same dimension, so the \pzpdf s are all defined over $z$, not $z_{i}$ as if it were another dimension unique to each galaxy $i$.
Rather, what differs between the \pzpdf s of the galaxy sample is the photometric data $\data_{i}$ upon which our knowledge of each galaxy's redshift is conditioned.
We correct our shorthand of ${\pr{z_{i}^{\dagger} = z'}}$ for the probability that a particular galaxy's true redshift takes a reference value with a more complete ${\pr{z = z' \gvn \data_{i}}}$ for the probability of a redshift being equal to a reference value conditional on the data observed from a particular galaxy in preparation to explore the consequences of this clarification.

The dependence of \pzpdf s on data transforms \Eq{eqn:stack} into
    \begin{align}
    \label{eqn:stackwithdata}
    \hat{\mathcal{N}}(z) &= \sum_{i = 1}^{N} \pr{z \gvn \data_{i}} .
    \end{align}
Crucially, \Eq{eqn:stackwithdata} is not correct. 
While it may look like Definition~\ref{def:normalization}, we can never integrate over the fixed value on the right side of a conditional; 
it is only mathematically valid to integrate over the free variable on the left side of a conditional \citep{hogg_data_2012}.

If the photometric data $\data_{i}$ is optimistically informative, meaning ${\pr{z \gvn \data_{i}} \approx \delta(z,\ z_{i}^{\dagger})}$, then \Eq{eqn:stackwithdata} approaches $K^{\dagger}$.
Though \pzpdf s cannot be delta functions because photometry is inherently noisy,
the PMF of an idealized errorless spectroscopic redshift could be considered a delta function like this.
We now assess whether \Eq{eqn:complete} yields an equivalent result to \Eq{eqn:stackwithdata} in the idealized case of errorless spectroscopic-quality observations.

\begin{theorem}
	\label{thm:informative}
	The mathematically valid expected value $\langle K \rangle$ of the redshift distribution at a reference redshift $z'$ given by \Eq{eqn:complete} is equivalent to the stacked estimator $\hat{\mathcal{N}}(z')$ at the reference redshift and the true number $K^{\dagger}$ of galaxies at the reference redshift if the data is perfectly informative, with ${\{\pr{z' \gvn \data_{i}} = \delta(z',\ z^{\dagger}_{i})\}}$.
\end{theorem}
\begin{proof}
	First, we replace all instances of ${\pr{z_{i}^{\dagger} = z'}}$ in \Eq{eqn:complete} with $\pr{z' \gvn \data_{i}}$, yielding
	    \begin{align}
	    \label{eqn:proofwithdata1}
	    \langle K \rangle &= \sum_{K' = 0}^{N} K' \sum_{j \backepsilon |S_{j}|=K'}^{\binom{N}{K'}} \left[ \prod_{i \in S_{j}} \pr{z' \gvn \data_{i}} \prod_{i \notin S_{j}} 1 - \pr{z' \gvn \data_{i}} \right] .
	    \end{align}
	For galaxies $i$ with perfectly informative data, each of their individual \pzpdf s ${\pr{z' \gvn \data_{i}}}$ becomes a delta function $\delta(z,\ z^{\dagger}_{i})$ centered at their true redshift, making \Eq{eqn:proofwithdata1} into
	    \begin{align}
	    \label{eqn:proofwithdata2}
	    \langle K \rangle &= \sum_{K' = 0}^{N} K' \sum_{j \backepsilon |S_{j}| = K'}^{\binom{N}{K'}} \left[ \prod_{i \in S_{j}} \delta(z', z^{\dagger}_{i}) \prod_{i \notin S_{j}} 1 - \delta(z', z^{\dagger}_{i}) \right] .
	    \end{align}
	Each of the delta function terms in \Eq{eqn:proofwithdata2} will be $1$ when ${z^{\dagger}_{i} = z'}$ and $0$ otherwise.
	Only sets of galaxies with ${|S_{j}| = K^{\dagger}}$ contribute to the outer sum, and there is only one set of galaxies $S_{j^{\dagger}}$ that results in a nonzero product within the inner sum.
	This single contributing set of galaxies is the one containing all ${K = K^{\dagger}}$ galaxies with ${z^{\dagger}_{i} = z'}$ and no interlopers with ${z^{\dagger}_{i} = \lnot z'}$.
	Thus the sole contributing term to the nested summations is
	    \begin{align}
	    \label{eqn:proofwithdata3}
	    \langle K \rangle &= K^{\dagger} \left[ \prod_{i \in S_{j^{\dagger}}} \delta(z', z^{\dagger}_{i}) \prod_{i \notin S_{j^{\dagger}}} 1 - \delta(z', z^{\dagger}_{i}) \right] .
	    \end{align}
	Since the entire bracketed quantity in \Eq{eqn:proofwithdata3} is $1$, we arrive at the desired ${\langle K \rangle = K^{\dagger}}$, the same as what stacking yields.
\end{proof}

To reiterate, the mathematically complete expected value of the redshift distribution is equivalent to the stacked estimator of the redshift distribution in the case of delta function \pzpdf s with perfect accuracy and precision, as would result from an idealized spectroscopic survey.
Ongoing and upcoming broad-band photometric surveys are unlikely to achieve the requisite photometric precision and accuracy for the validity of the stacked estimator.
If the \pzpdf s are not delta functions but can be well-described by a specific parametric form, it may also be possible to recover the true redshift distribution via deconvolution \citep{padmanabhan_calibrating_2005}.
Though it is outside the scope of this letter, a quantification of the information loss, on $N(z)$ and on recovered cosmological parameter constraints, due to using the stacked estimator rather than a mathematically self-consistent estimator under a generic parameterization of \lsst-like \pz\ requirements can be found in \cite{malz_how_2020}.

\subsection{Perfectly uninformative data}
\label{sec:uninformative}

\Sect{sec:informative} concerns \pzpdf s perfectly informed by  photometry, but that picture is still overly simplistic.
If \pzpdf s were conditioned solely on photometry, there would be no disagreement about how to derive them from photometric catalogs, because every approach would yield the same result. 
Thus the differences between \pzpdf s derived by each of dozens of methods indicate that \pzpdf s must be conditioned not only on data unique to each galaxy but also on prior information $\tilde{\theta}$ corresponding to a model for the relationship between data $\data$ and redshift $z$. 
This prior information may come from a template library or training set, as well as the way each algorithm combines those pieces of information with the actual observations to arrive at an estimated \pzpdf\ \citep{schmidt_evaluation_2020}.
The prior information is projected into the space of $\mathcal{N}(z')$ as some interim guess $\tilde{K}$ for the number of galaxies with reference redshift $z'$.

If \pzpdf s are more completely written as ${\pr{z \gvn \data_{i}, \tilde{K}}}$, \Eq{eqn:stackwithdata} becomes
    \begin{align}
    \label{eqn:stackwithprior}
    \hat{\mathcal{N}}(z) &= \sum_{i = 1}^{N} \pr{z \gvn \data_{i}, \tilde{K}} .
    \end{align}
If the photometric data were totally uninformative, as could occur  pessimistically in an infinitely deep, completely noise-dominated photometric survey, each galaxy $i$ would have the same ${\pr{z \gvn \data_{i}, \tilde{K}} \approx \pr{z \gvn \tilde{K}}}$. 
Furthermore, though the prior $\tilde{K}$ may take any possible value for $K$,  it is baked into the \pzpdf s on the right side of the conditional and cannot in general be modified at will.
Now, we consider what happens to \Eq{eqn:complete} under perfectly uninformative data with a prior $\tilde{K}$.

\begin{theorem}
	\label{thm:uninformative}
	The mathematically valid expected value $\langle K \rangle$ of the number of galaxies with reference redshift $z'$ is equivalent to the stacked estimator $\hat{\mathcal{N}}(z')$ at the reference redshift; both $\langle K \rangle$ and $\hat{\mathcal{N}}(z')$ are equal to the true number $K^{\dagger}$ of galaxies with redshift $z'$ under a special case of perfectly uninformative photometry and perfect prior information ${\tilde{K} = K^{\dagger}}$.
\end{theorem}
\begin{proof}
	If every galaxy has the same \pzpdf, then every galaxy also has the same probability ${\pr{z' \gvn \tilde{K}} = \frac{\tilde{K}}{N}}$ of having the reference redshift.
	The stacked estimator is thus ${\hat{\mathcal{N}}(z) = N \cdot \pr{z \gvn \tilde{K}} = \tilde{K}}$.
	To start deriving the complete, correct result, we first update \Eq{eqn:complete} to reflect our new understanding of the role of the prior $\tilde{K}$ in estimated \pzpdf s.
	    \begin{align}
	    \label{eqn:proofwithprior1}
	    \langle K \rangle &= \sum_{K' = 0}^{N} K' \sum_{j \backepsilon |S_{j}|=K'}^{\binom{N}{K'}} \left[ \prod_{i \in S_{j}} \pr{z' \gvn \tilde{K}} \prod_{i \notin S_{j}} 1 - \pr{z' \gvn \tilde{K}} \right] .
	    \end{align}
	Without having to worry about the true galaxy redshifts $z^{\dagger}_{i}$, we can rewrite this in the form of \Eq{eqn:binomial} as
	    \begin{align}
	    \label{eqn:proofwithprior2}
	    \langle K \rangle &= \sum_{K' = 0}^{N} K' \sum_{j \backepsilon |S_{j}| = K'}^{\binom{N}{K'}} \left[ \left(\frac{\tilde{K}}{N}\right)^{K'} \left(1 - \frac{\tilde{K}}{N}\right)^{N - K'} \right] .
	    \end{align}
	Since the \pzpdf s are identical to one another, the bracketed terms are the same for a given $K'$, meaning each term in the inner sum is the same, so we can eliminate the inner sum as 
	    \begin{align}
	    \label{eqn:proofwithprior3}
	    \langle K \rangle &= \sum_{K' = 0}^{N} K' \binom{N}{K'} \left(\frac{\tilde{K}}{N}\right)^{K'} \left(1 - \frac{\tilde{K}}{N}\right)^{N - K'} .
	    \end{align}
	Noting that the first term with $K = 0$ always vanishes, canceling factors of $K$, and factoring out a quantity motivated by the coin-flipping analogy of \Sect{sec:tworedshift}, we  rewrite \Eq{eqn:proofwithprior3} as 
	    \begin{align}
	    \label{eqn:proofwithprior4}
	    \langle K \rangle &= N \frac{\tilde{K}}{N} \sum_{K' = 1}^{N} \binom{N - 1}{K' - 1} \left(\frac{\tilde{K}}{N}\right)^{K' - 1} \left(1 - \frac{\tilde{K}}{N}\right)^{(N - 1) - (K' - 1)} .
	    \end{align}
	We recognize this as a binomial expansion of
	\begin{align}
	    \label{eqn:proofwithprior5}
	    \langle K \rangle &= \tilde{K} \left(\frac{\tilde{K}}{N} + 1 - \frac{\tilde{K}}{N}\right)^{N - 1} .
	    \end{align}
	Since the term in parentheses is $1$, we arrive at ${\langle K \rangle = \tilde{K}}$ for this case.
	As $\langle K \rangle = \tilde{K} = \hat{\mathcal{N}}(z')$, the only way for the result of stacking and the expected value of the redshift distribution to be equal to the true value $K^{\dagger}$ of the redshift distribution is if $\tilde{K} = K^{\dagger}$, i.e. if the prior is equal to the truth.
\end{proof}

Thus stacking is valid if the data is uninformative and if the true $N(z)$ is known \textit{a priori}.
Previous low-redshift galaxy surveys with a sufficient sample size to overcome cosmic variance could serve as an appropriate prior for one another's redshift distributions because they would have complete support over the redshift range of their galaxy samples.
However, at higher redshifts or in atypical galaxy environments, a lucky guess of ${\tilde{K} \approx K^{\dagger}}$ is tremendously unlikely;
in fact, if the truth were known independently of the redshift survey itself, it would not be necessary to collect the photometric data at all. 
Though stacking works in the case of uninformative data with a prior equal to the truth, this situation is not realistic for upcoming nor ongoing galaxy surveys.

\section{Conclusion}
\label{sec:disco}

This letter explores the common assumptions and misconceptions that have inadvertently disguised the stacked estimator of the redshift distribution as valid throughout the cosmological literature.
We establish a complete and consistent way to frame the redshift distribution \Nz\ in terms of \pz\ PDFs and provide proofs comparing the stacked estimator to the provably correct but computationally prohibitive expected value of the distribution of possible redshift distributions.
We find that equivalence is only achievable in two idealized limiting cases, thereby addressing the dual questions of how stacking earned its place among established methodologies in observational cosmology and why it can no longer be trusted in future applications.

The only cases in which stacking is able to recover the true \Nz\ are when the \pz\ PDFs are conditioned on perfect data or conditioned on perfect prior information.
The former case may have been approximately true for high-fidelity spectroscopic or narrow-band photometric redshift surveys of bright sources, however it is unlikely to hold for broad-band photometric surveys probing faint galaxies, such as those used in cosmological weak lensing.
The latter case may have been approximately true for galaxy surveys whose prior-inspiring predecessors covered the same redshift range with comparable depth and sufficient area to overcome cosmic variance, but upcoming surveys will explore higher redshift ranges than have been previously studied, violating the perfect prior condition. 
Finally, while it is possible to imagine a circumstance in which realistically informative and uninformative \pz\ PDFs and priors conspire to produce a stacked \Nz\ that faithfully recovers the truth, such a solution would necessarily have to be very finely tuned and thus is not generically achievable.

While we do not address the numerical degree to which the requirements of stacking were met by previous and ongoing studies, we argue that they will certainly not be met in future cosmological studies. 
Therefore, the message is clear; 
for the accurate estimation of the redshift distribution in next-generation broad-band photometric surveys, we must not stack!

\section*{Acknowledgments}

I thank Michael Blanton, Johann Cohen-Tanugi, Daniel Gruen, Alan Heavens, David Hogg, Jeff Kuan, Rachel Mandelbaum, Eduardo Rozo, Ted Singer, Anze Slosar, James Stevenson, and Angus Wright for helpful feedback on this letter as it developed.
I further thank all of the above, as well as David Alonso, Will Hartley, Mike Jarvis, Boris Leistedt, Tom Loredo, Mark Manera, Phil Marshall, Jeff Newman, and Josh Speagle, for the conversations that inspired this work.
Finally, I express sincere gratitude to the anonymous referee for suggestions that significantly improved this manuscript.

AIM acknowledges support from the Max Planck Society and the Alexander von Humboldt Foundation in the framework of the Max Planck-Humboldt Research Award endowed by the Federal Ministry of Education and Research.
During the completion of this work, AIM was advised by David W. Hogg and supported by National Science Foundation grant AST-1517237.
This work was initiated while AIM was supported by the U.S. Department of Energy, Office of Science, Office of Workforce Development for Teachers and Scientists, Office of Science Graduate Student Research (SCGSR) program, administered by the Oak Ridge Institute for Science and Education for the DOE under contract number DE‐SC0014664.

\bibliography{pedant}

\end{document}